\documentclass[journal]{IEEEtran}
\usepackage{stfloats}
\usepackage{float}
\usepackage{graphicx}
\usepackage{subfigure}
\usepackage{color}
\usepackage{algpseudocode}
\usepackage{amssymb}
\usepackage{amsmath}
\usepackage{algorithm}
\usepackage{bm}
\usepackage[colorlinks,citecolor=blue,linkcolor=blue,urlcolor=blue]{hyperref}

\newtheorem{remark}{Remark}
\newtheorem{theorem}{Theorem}
\newtheorem{corollary}{Corollary}

\newtheorem{proof}{Proof}
\newtheorem{proposition}{Proposition}

\begin{document}

	\title{Physical Layer Security for FAS-Aided Short-Packet Systems: A Variable Block-Correlation Approach}
	
	\author{Jianchao Zheng, Tuo Wu, Kai-Kit Wong,~\IEEEmembership{Fellow,~IEEE},  Baiyang Liu,  Runyu Pan, Maged Elkashlan, \\ Kin-Fai Tong,~\IEEEmembership{Fellow,~IEEE}, Hyundong Shin,~\IEEEmembership{Fellow,~IEEE}, and Sumei Sun,~\IEEEmembership{Fellow,~IEEE} 
		
		\thanks{(\textit{Corresponding author: Tuo Wu.})}
		\thanks{J. Zheng is with the School of Computer Science and Engineering, Huizhou University, Huizhou 516007, China (E-mail: $\rm zhengjch@hzu.edu.cn$).}
		\thanks{T. Wu is with the Department of Electrical Engineering, City University of Hong Kong, Hong Kong, China (E-mail: $\rm tuowu2@cityu.edu.hk$).}
		\thanks{R. Pan is with the School of Computer Science and Technology, Shandong University, Qingdao 266237, China (E-mail: $\rm rypan@sdu.edu.cn$).}
		\thanks{K.-K. Wong is with the Department of Electronic and Electrical Engineering, University College London, WC1E7JE London, U.K., and is also affiliated with the Department of Electronic Engineering, Kyung Hee University, Yongin-si, Gyeonggi-do 17104, Korea. (E-mail: $\rm kai$-$\rm kit.wong@ucl.ac.uk$).} 
		\thanks{B. Liu and K. F. Tong are with the School of Science and Technology, Hong Kong Metropolitan University, Hong Kong SAR, China (E-mail: $\rm \{byliu, ktong\}@hkmu.edu.hk$).}
		\thanks{M. Elkashlan is with Queen Mary University of London, United Kingdom (E-mail: $\rm maged.elkashlan@qmul.ac.uk$).}
		\thanks{H. Shin is with the Department of Electronics and information Convergence Engineering, Kyung Hee University, Yongin-si.Gyeonggi-do 17104, Republic of Korea (E-mail: $\rm hshin@khu.ac.kr$).}
		\thanks{S. Sun is with the Institute for Infocomm Research (I2R), A*STAR, Singapore (E-mail: $\rm sunsm@a$-$\rm star.edu.sg$).}
	}
	\markboth{}
	{Wu \MakeLowercase{\textit{et al.}}: Physical Layer Security for FAS-Aided Short-Packet Systems: A Variable Block-Correlation Approach}
	
	\maketitle
	
	\begin{abstract}
	This paper presents a comprehensive physical layer security (PLS) framework for fluid antenna system (FAS)-aided short-packet communications under the variable block-correlation model (VBCM). We consider a downlink wiretap scenario in which a base station transmits confidential short packets to a legitimate receiver user (RU) in the presence of an eavesdropper user (EU), where both the RU and EU are equipped with fluid antennas. Unlike existing FAS security analyses that rely on constant block-correlation models or infinite-blocklength assumptions, we incorporate the VBCM to accurately capture the non-uniform spatial correlation structure inherent in practical FAS deployments. By employing a piecewise linear approximation of the decoding error probability and Gauss-Chebyshev quadrature, we derive closed-form and asymptotic expressions for the average achievable secrecy throughput (AAST). We further prove that the AAST is monotonically non-decreasing in the number of RU ports, which reduces the three-dimensional joint optimization of transmit power, blocklength, and port number to a two-dimensional grid search (GS). Numerical results demonstrate that the FAS-aided system achieves up to an order-of-magnitude secrecy throughput improvement over conventional fixed-position antenna systems, and reveal that blocklength selection is the most critical design parameter in the joint optimization.
	\end{abstract}
	
	
	\section{Introduction}

	{\IEEEPARstart{T}{he} evolution of wireless communications towards sixth-generation (6G) networks demands cutting-edge technologies that simultaneously deliver ultra-high reliability, low latency, and stringent security guarantees \cite{Tariq-2020}. Among the transformative applications envisioned for 6G, ultra-reliable low-latency communications (URLLC) and mission-critical Internet-of-Things (IoT) rely fundamentally on short-packet transmissions, where finite blocklength effects introduce an irreducible decoding error probability that departs significantly from the classical Shannon capacity \cite{Polyanskiy20}. In addition, advanced enabling techniques such as high-precision localization \cite{TWuJSAC24, TWuIoTJ24, TWuCL24}, intelligent resource management \cite{TWuTMC25}, and integrated sensing and communication (ISAC) \cite{MaD26} are also pivotal to realizing 6G visions. As these short-packet applications proliferate in security-sensitive domains---including autonomous vehicles, remote surgery, and industrial control---ensuring the confidentiality of short-packet transmissions has emerged as a critical and urgent research challenge.
		
	Massive multiple-input multiple-output (MIMO) serves as a cornerstone of contemporary fifth-generation (5G) networks, exploiting a large number of antennas for spatial diversity \cite{FuJ26}. However, its performance is fundamentally limited by the number of available radio-frequency (RF) chains at the terminals, while rising hardware costs, power consumption, and operational complexity pose substantial scalability challenges. Fluid antenna systems (FAS) \cite{NewW25,LuW25,HongH26,NewW26,TWu20243},  sometimes referred to movable antennas  \cite{Zhu-Wong-2024,LZhu25}, have emerged as a paradigm-shifting solution that overcomes these limitations \cite{YaoJ241, HXu23, XLai23, BC24, LaiX242, YaoJ251, JYao2024, YaoJ252, Ghadi-2023, New-twc2023, NWaqar23, CWang24, TWuTCOMM26, TWuTVT25, TWuJSTSP25, YaoJTWC26, HChenTNSE25, HXuTWC25, LZhouWCL24, YaoJTMC25}. The broader FAS framework encompasses electronically switchable arrays \cite{FAS21,WongK20}, reconfigurable metasurfaces \cite{BLiu25}, and spatially distributed antenna pixels \cite{Zhang25}, enabling dynamic control of spatial radiating points to enhance the degrees of freedom (DoFs) far beyond conventional fixed-position antenna (FPA) arrays \cite{TWu20243}.
	
	Physical layer security (PLS) has received considerable attention as a key technique for securing wireless communications without relying on cryptographic keys \cite{Wyner1975,HNiututorial}. Within the FAS paradigm, PLS introduces both distinct challenges and novel opportunities: the intelligent port selection mechanism of FAS enables the legitimate user to exploit spatial diversity for secrecy enhancement, while the eavesdropper may also benefit from similar reconfigurability. These aspects have been investigated in several recent studies \cite{Shojaeifard, Wong-frontiers22,Security1,Security2,Security3,Security4,Security5,SYangWCL25,TuoW}. Nonetheless, existing FAS security analyses predominantly operate under two simplifying assumptions: (i) infinite blocklength (Shannon capacity), which overlooks the finite blocklength effects inherent in short-packet transmissions, and (ii) constant block-correlation models (BCM), which assume uniform correlation coefficients across all channel blocks.
	
	The accuracy of spatial correlation modeling is critical for FAS security analysis, as the correlation structure directly governs the channel statistics of both the legitimate receiver and the eavesdropper, thereby determining the secrecy capacity and outage performance. The $N\times N$ channel covariance matrix of FAS exhibits a Toeplitz structure that renders exact analytical treatment intractable \cite{FAS22}. To enable tractable analysis, spatial block-correlation models (BCM) have been proposed, which partition the Toeplitz-structured covariance matrix into $D$ independent blocks, each characterized by a correlation coefficient \cite{BC24, LaiX242}. While the original work \cite{BC24} recognized that allowing correlation coefficients to vary across different blocks could in principle yield better approximation accuracy, it focused on the practically simpler constant BCM---where a single uniform coefficient is shared by all blocks---to avoid the exponential $D^{N-D}$ search complexity of the general variable case. However, the constant BCM suffers from fundamental limitations, especially in compact FAS deployments. When the number of ports $N$ is small (e.g., $N < 20$), the eigenvalue distribution of the Toeplitz covariance matrix lacks significant sparsity, meaning that the differences between individual eigenvalues are not negligible and each port substantially influences the overall channel statistics. Forcing a uniform correlation coefficient across all blocks in such settings fails to capture the non-uniform intra-block correlation structure, leading to significant approximation errors that are particularly detrimental in security-sensitive applications where modeling inaccuracies translate directly to over- or under-estimation of the secrecy throughput \cite{FAS22, Chai22, FAMS}. To address this fundamental limitation, the variable block-correlation model (VBCM) was recently established in \cite{LaiX}, which allows the correlation coefficient to vary across different blocks. By leveraging rigorous eigenvalue analysis, \cite{LaiX} derives closed-form expressions for the optimal block-specific correlation coefficients and develops a low-complexity algorithm that reduces the search complexity from the exponential $D^{N-D}$ to a linear $(N-D)\times D$ procedure, achieving substantially lower approximation error compared to the constant BCM.
	
	The convergence of FAS, short-packet communications, and PLS creates a compelling yet largely unexplored research direction. From a practical standpoint, many security-critical 6G applications---such as URLLC-based industrial control, autonomous driving, and remote healthcare---inherently require short-packet transmissions, making the finite blocklength regime the natural operating condition rather than an edge case. Moreover, FAS is uniquely suited to enhance short-packet security: the port selection diversity of FAS provides a spatial DoF that can significantly reduce the decoding error probability at the legitimate receiver while maintaining a statistical advantage over the eavesdropper. Unlike the infinite blocklength regime, where secrecy is determined solely by the signal-to-noise ratio (SNR) gap, the short-packet regime introduces a channel dispersion penalty $\sqrt{V_R/M}$ that amplifies the sensitivity of the secrecy throughput to the channel quality---making the diversity gain from FAS port selection even more valuable. Furthermore, the VBCM is essential (not merely beneficial) for this analysis. In particular, when the number of antenna ports is limited (typically $N < 20$), the constant BCM's inability to capture non-uniform intra-block correlations can produce substantial errors in the predicted secrecy throughput. Overly optimistic models may lead to dangerously under-designed systems, while overly pessimistic models result in unnecessarily conservative---and therefore inefficient---resource allocation. By assigning block-optimal correlation coefficients derived from eigenvalue analysis \cite{LaiX}, the VBCM provides a reliable and tractable characterization of the correlated FAS channel that is indispensable for accurate short-packet security design \cite{BC24, FAS22}.
	
	Despite the strong motivation, the integration of VBCM into FAS short-packet security analysis poses formidable analytical challenges. \textit{First}, the VBCM's multi-block cumulative distribution function (CDF) product structure, combined with the Gaussian Q-function from finite blocklength theory, results in complex nested integrals over Marcum Q-functions and multi-dimensional probability distributions that do not admit closed-form solutions. \textit{Second}, the piecewise linear approximation of the decoding error probability partitions the integration domain into multiple regions whose boundaries are implicit functions of the eavesdropper's SNR, requiring careful case-by-case analysis and numerical root-finding. \textit{Third}, the joint optimization of transmit power $P$, blocklength $M$, and port number $N_R$ constitutes a mixed-integer non-convex program, where $N_R$ is discrete and the average achievable secrecy throughput (AAST) is non-monotonic in both $M$ and the number of information bits $m$, giving rise to a complex multi-modal optimization landscape.

	To overcome these challenges, this paper develops a comprehensive analytical and optimization framework for FAS-aided secure short-packet communications under the VBCM. Our key approach involves: (i) deriving closed-form AAST expressions by combining the piecewise linear approximation of the decoding error probability with Gauss-Chebyshev quadrature for efficient numerical integration; (ii) obtaining asymptotic AAST expressions in the high-SNR regime that provide design insights and enable low-complexity optimization; and (iii) proving a monotonicity property of the AAST with respect to $N_R$, which reduces the three-dimensional optimization to a two-dimensional grid search (GS). Our framework fundamentally differs from prior FAS security works \cite{Security1,Security2,Security3,Security4,Security5} in three aspects: it replaces the constant BCM with the more accurate VBCM, it operates in the practically relevant finite blocklength regime rather than the Shannon limit, and it jointly optimizes all three key system parameters with provable convergence guarantees.}

The primary contributions of this work are summarized as follows:

\begin{itemize}
	
	\item \textbf{\textit{VBCM-Based Security Analysis Framework:}} We establish the first analytical framework for FAS-aided secure short-packet communications under the VBCM. By employing a piecewise linear approximation of the decoding error probability and Gauss-Chebyshev quadrature, we derive closed-form expressions for the AAST that jointly account for the variable spatial correlation structure and finite blocklength effects. Asymptotic AAST expressions are further obtained in the high-SNR regime, providing tractable design guidelines.
	
	\item \textbf{\textit{Structural Properties and Dimensionality Reduction:}} We prove that the asymptotic AAST is monotonically non-decreasing in $N_R$ (Proposition~\ref{prop:monotone}), which yields the corollary that the optimal port number is always $N_R^* = N_{\max}$. This structural result reduces the original three-dimensional mixed-integer optimization to a two-dimensional continuous search, significantly lowering computational complexity.
	
	\item \textbf{\textit{Joint Optimization with Convergence Guarantees:}} We develop a GS optimization algorithm to jointly optimize the transmit power, blocklength, and port number for maximizing the secrecy throughput. The algorithm is equipped with a theoretical convergence analysis based on Lipschitz continuity, providing explicit error bounds and minimum grid resolution requirements.
	
	\item \textbf{\textit{Comprehensive Design Insights:}} Through extensive numerical analysis, we reveal that (i) FAS achieves up to an order-of-magnitude secrecy throughput gain over conventional FPA systems, (ii) blocklength selection is the most critical design parameter, (iii) the secrecy gain is governed by the port count asymmetry $N_R/N_E$ rather than absolute port numbers, and (iv) FAS port diversity is most effective in the moderate-SNR regime.

\end{itemize}

	The remainder of this paper is organized as follows. Section~II introduces the system model for FAS-enabled secure communication and the channel correlation model. Section~III presents the security analysis framework, including the derivation of the AAST and asymptotic AAST expressions. Section~IV details the optimization algorithm for FAS security enhancement. Section~V provides numerical results and performance validation. Finally, Section~VI concludes the paper.
	
	\section{System Model}
	
	We consider a wireless downlink secure short-packet communication system, in which a single-antenna base station (BS) transmits confidential information to a legitimate receiver user (RU) in the presence of an eavesdropper user (EU). Both the RU and the EU are equipped with one-dimensional linear fluid antennas capable of switching among $N_R$ and $N_E$ discrete ports, respectively, within a linear space of $W\lambda$, where $W$ denotes the normalized antenna size and $\lambda$ is the carrier wavelength. We assume quasi-static block fading channels, where the channel coefficients remain constant within each transmission block but vary independently across different blocks.
	
	\subsection{Signal Model}
	The BS transmits the signal $s\sim \mathcal{CN}(0, 1)$ to the RU with transmit power $P$. The signals received at the $k$-th port of the RU and the $j$-th port of the EU are respectively given by
	\begin{align}
		y_{R,k} &= \sqrt{P} g_{R,k} d_{R}^{-\frac{a_R}{2}}s + n_{R,k}, \quad k \in \{1, \dots, N_{R}\},\\
		y_{E,j} &= \sqrt{P} g_{E,j} d_{E}^{-\frac{a_E}{2}}s + n_{E,j}, \quad j \in \{1, \dots, N_{E}\},
	\end{align}
	where $g_{R,k} \sim \mathcal{CN}(0, \eta_{R})$ and $g_{E,j} \sim \mathcal{CN}(0, \eta_{E})$ denote the complex Gaussian channel coefficients from the BS to the $k$-th port of the RU and the $j$-th port of the EU, respectively, with $\eta_{R}$ and $\eta_{E}$ representing the corresponding average channel power gains. In addition, $n_{R,k} \sim \mathcal{CN}(0, \sigma_{R}^2)$ and $n_{E,j} \sim \mathcal{CN}(0, \sigma_{E}^2)$ represent the additive white Gaussian noise (AWGN) at the RU and the EU, respectively. Moreover, $d_{R}$ and $d_{E}$ denote the distances from the BS to the RU and the EU, respectively, and $a_R$ and $a_E$ are the path-loss exponents for the RU and the EU links, respectively.
	
	By exploiting the spatial diversity offered by the fluid antenna, both users select the optimal port that maximizes the channel amplitude gain, i.e.,
	\begin{align}
		i^* &= \arg\max_{k \in \{1, \dots, N_R\}} |g_{R,k}|, \nonumber\\
		j^* &= \arg\max_{j \in \{1, \dots, N_E\}} |g_{E,j}|.
	\end{align}
	Accordingly, the maximum channel amplitude gains at the RU and the EU are defined as $\alpha_{R} = |g_{R,i^*}|$ and $\alpha_{E} = |g_{E,j^*}|$, respectively. The instantaneous signal-to-noise ratios (SNRs) at the selected ports of the RU and the EU are then given by
	\begin{align}
		\gamma_R = \frac{P d_R^{-a_R}}{\sigma_R^2} \alpha_R^2, \quad \gamma_E = \frac{P d_E^{-a_E}}{\sigma_E^2} \alpha_E^2. \label{q6}
	\end{align}
	
	\subsection{FAS Channel Correlation Model}
	Due to the small spacing between adjacent ports of the fluid antenna, the channel coefficients are spatially correlated. The channel vector of the RU is defined as $\mathbf{g}_R = [g_{R,1}, \dots, g_{R,N_R}]^T \sim \mathcal{CN}(\mathbf{0}, \mathbf{\Sigma}_R)$, where $\mathbf{\Sigma}_R = \mathbb{E}[\mathbf{g}_R \mathbf{g}_R^H]$ is the covariance matrix. Following the Jakes model \cite{FAS21}, the $(k,l)$-th entry of the normalized correlation matrix is given by
	\begin{align}
		[\mathbf{R}]_{k,l} = J_0\left(\frac{2\pi |k-l| W}{N_R-1}\right),
	\end{align}
	where $J_0(\cdot)$ denotes the zero-order Bessel function of the first kind. Thus, the covariance matrix of the RU can be expressed as
	\begin{align}
		\mathbf{\Sigma}_R = \eta_R \begin{pmatrix}
			1 & [\mathbf{R}]_{1,2} & \cdots & [\mathbf{R}]_{1,N_R} \\
			[\mathbf{R}]_{1,2} & 1 & \cdots & [\mathbf{R}]_{2,N_R} \\
			\vdots &  & \ddots & \vdots \\
			[\mathbf{R}]_{1,N_R}& [\mathbf{R}]_{2,N_R} & \cdots & 1
		\end{pmatrix},
	\end{align}
	which exhibits a symmetric Toeplitz structure. Similarly, the channel vector of the EU is defined as $\mathbf{g}_E = [g_{E,1}, \dots, g_{E,N_E}]^T \sim \mathcal{CN}(\mathbf{0}, \mathbf{\Sigma}_E)$, where $\mathbf{\Sigma}_E = \eta_E \mathbf{R}_E$ and $[\mathbf{R}_E]_{i,j} = J_0\left(\frac{2\pi |i-j| W}{N_E-1}\right)$.

	\section{Security Throughput of FAS Communication}
	In this section, we derive the closed-form and asymptotic expressions for the average achievable secrecy throughput (AAST) of the considered FAS-aided secure short-packet communication system.
			
	\subsection{Achievable Secrecy Rate and Performance Metrics Under Finite Blocklength}
	Based on the SNR expressions in \eqref{q6}, the achievable secrecy rate under finite blocklength can be approximated as \cite{Security5}
	\begin{align}
	R(M,\epsilon,\delta)=C_{S}-\sqrt{\frac{V_{R}}{M}}\frac{Q^{-1}(\epsilon)}{\ln 2}-\sqrt{\frac{V_{E}}{M}}\frac{Q^{-1}(\delta)}{\ln 2},
	\end{align} 
	conditioned on $\gamma_R>\gamma_E$ (otherwise the secrecy rate is zero), where $M$ denotes the blocklength, $\epsilon$ is the decoding error probability at the RU, and $\delta$ represents the information leakage level. Here, $C_s=\log_2(1 + \gamma_R)-\log_2(1 + \gamma_E)$ is the secrecy capacity, $V_{R} = 1 - (1 + \gamma_{R})^{-2}$ and $V_{E} = 1 - (1 + \gamma_{E})^{-2}$ are the channel dispersions at the RU and the EU, respectively, and $Q^{-1}(\cdot)$ denotes the inverse of the Gaussian Q-function $Q(x) = \int_{x}^{\infty} \frac{1}{\sqrt{2\pi}} e^{-\frac{t^{2}}{2}} dt$.
	
	The BS transmits $m$ information bits to the RU. Given the blocklength $M$, the transmission rate is $R=m/M$. The decoding error probability $\epsilon$ is given by
	\begin{align}
	\epsilon=Q\left(\sqrt{\frac{M}{V_{R}}}\left(\ln\frac{1+\gamma_{R}}{1+\gamma_{E}}-\sqrt{\frac{V_{E}}{M}}Q^{-1}(\delta)-\frac{m}{M}\ln 2\right)\right),
	\end{align}
	for $\gamma_R>\gamma_E$. When $\gamma_R\leq\gamma_E$, the secrecy capacity is zero and we simply set $\epsilon=1$ for this trivial case. The average achievable secrecy throughput  can be expressed as
	\begin{align}\label{eq:st_def}
	T=\mathbb{E}_{\gamma_{R},\gamma_{E}}\left[\frac{m}{M}(1-\epsilon)\right]=\frac{m}{M}\left(1-\mathbb{E}_{\gamma_{R},\gamma_{E}}[\epsilon]\right),
	\end{align}
	where $\mathbb{E}_{\gamma_{R},\gamma_{E}}[\epsilon]$ denotes the average decoding error probability.
	
	\subsection{Variable Block-Correlation Model and Channel Statistics Distribution}
	
	To evaluate \eqref{eq:st_def}, the statistical distributions of $\gamma_{R}$ and $\gamma_{E}$ are required. However, direct analysis based on the Toeplitz covariance matrix $\mathbf{\Sigma}$ is analytically intractable. To address this, we follow \cite{TuoW,LaiX} and employ the variable block-correlation model (VBCM) as a mathematically tractable approximation. Specifically, the VBCM partitions the $N\times N$ Toeplitz covariance matrix into $D$ independent blocks, where, unlike the conventional constant BCM \cite{BC24} that assigns a single uniform coefficient to all blocks, the VBCM assigns each block $d$ its own optimal correlation coefficient $\rho_d$ derived from closed-form eigenvalue analysis \cite{LaiX}. This block-specific parameterization accurately captures the non-uniform spatial correlation structure, reducing approximation error from the exponential $D^{N-D}$ brute-force search to a linear $(N-D)\times D$ procedure \cite{LaiX}. Under the VBCM framework \cite{TuoW}, the overall CDF of the user's maximum channel amplitude is expressed as the product of individual block CDFs, i.e.,
	\begin{align}\label{eq:cdf_final}
		F_{\alpha}(x) = \prod_{d=1}^{D} F_{\alpha_d}(x),
	\end{align}
	where 
	\begin{align}\label{eq:block_cdf}
		F_{\alpha_d}(x) = \int_0^\infty \left[1 - Q_1\left(\frac{\theta}{\sigma_d}, \frac{x}{\sigma_d}\right)\right]^{L_d} \frac{2\theta}{\eta\rho_d} e^{-\frac{\theta^2}{\eta\rho_d}} d\theta.
	\end{align}
	Here, $Q_1(\cdot,\cdot)$ denotes the first-order Marcum Q-function \cite{book}, $\rho_d \in [0, 1]$ is the variable correlation coefficient for block $d$, $\eta$ is the average channel power gain, $D$ is the number of independent blocks with sizes $\{L_d\}_{d=1}^D$ satisfying $\sum_{d=1}^D L_d = N$, and $\sigma_d = \sqrt{\eta(1-\rho_d)/2}$.
	
	The corresponding PDF is obtained by differentiating \eqref{eq:cdf_final} with respect to $x$, yielding
	\begin{align}\label{eq:pdf_final}
		f_{\alpha}(x) = \sum_{d=1}^{D} \left[ f_{\alpha_d}(x) \prod_{j=1, j \neq d}^{D} F_{\alpha_j}(x) \right],
	\end{align}
	where $f_{\alpha_d}(x) = \frac{d}{dx} F_{\alpha_d}(x)$ represents the PDF of the maximum amplitude in block $d$.
	
	\subsection{SNR Statistics Derivation}
	
	Based on the VBCM channel statistics derived above, we now obtain the CDF and PDF of the SNRs at the RU and the EU. For notational convenience, we define the average SNRs as $\bar{\gamma}_R \triangleq P d_R^{-a_R}/\sigma_R^2$ and $\bar{\gamma}_E \triangleq P d_E^{-a_E}/\sigma_E^2$, so that $\gamma_R = \bar{\gamma}_R \alpha_R^2$ and $\gamma_E = \bar{\gamma}_E \alpha_E^2$ from \eqref{q6}.
	
	The CDF of $\gamma_R$ is obtained by a direct transformation as
	\begin{align}\label{eq:cdf_snr_R}
	F_{\gamma_R}(z)=\Pr\left( \bar{\gamma}_R \alpha_R^2\leq z\right) = F_{\alpha_R}\left(\sqrt{\frac{z}{\bar{\gamma}_R}}\right).
	\end{align}
	Differentiating \eqref{eq:cdf_snr_R} with respect to $z$ via the chain rule yields the PDF of $\gamma_R$:
	\begin{align}\label{eq:pdf_snr_R}
	f_{\gamma_R}(z)=\frac{d}{dz}F_{\alpha_R}\left(\sqrt{\frac{z}{\bar{\gamma}_R}}\right) = \frac{1}{2\sqrt{\bar{\gamma}_R\, z}}\, f_{\alpha_R}\left(\sqrt{\frac{z}{\bar{\gamma}_R}}\right).
	\end{align}
	Similarly, the CDF of $\gamma_E$ is given by
	\begin{align}\label{eq:cdf_snr_E}
	F_{\gamma_E}(y) = F_{\alpha_E}\left(\sqrt{\frac{y}{\bar{\gamma}_E}}\right),
	\end{align}
	and the corresponding PDF is
	\begin{align}\label{eq:pdf_snr_E}
	f_{\gamma_E}(y)=\frac{1}{2\sqrt{\bar{\gamma}_E\, y}}\, f_{\alpha_E}\left(\sqrt{\frac{y}{\bar{\gamma}_E}}\right).
	\end{align}
	
	\begin{remark}[FAS Diversity and Security Implications]\label{remark:diversity}
	The derived SNR statistics in \eqref{eq:cdf_snr_R}--\eqref{eq:pdf_snr_E} reveal the fundamental role of FAS in secure communications. From \eqref{eq:cdf_final}, the CDF $F_{\alpha_R}(x) = \prod_{d=1}^{D} F_{\alpha_d}(x)$ is a product of $D$ block CDFs, each satisfying $F_{\alpha_d}(x) \leq 1$. As the number of RU ports $N_R$ increases (and correspondingly $D$ or $L_d$), $F_{\alpha_R}(x)$ decreases for a given $x$, which shifts the SNR distribution $F_{\gamma_R}(z)$ toward higher values. This is equivalent to a spatial diversity gain achieved through port selection, even with a single RF chain. In contrast, increasing $N_E$ shifts $F_{\gamma_E}(y)$ toward higher SNR values for the eavesdropper, degrading the secrecy performance. Consequently, the secrecy throughput is governed by the \emph{asymmetry} between the legitimate and eavesdropping channels: a favorable secrecy condition requires $\bar{\gamma}_R / \bar{\gamma}_E \gg 1$ or $N_R \gg N_E$ to ensure that $\gamma_R$ statistically dominates $\gamma_E$. This insight provides a design guideline that the RU should be equipped with significantly more FAS ports than the EU to guarantee a positive secrecy throughput.
	\end{remark}
	
	\subsection{Derivation of Secrecy Throughput}
	
	In this subsection, we derive the AAST defined in \eqref{eq:st_def}. Since $\epsilon = 1$ when $\gamma_R \leq \gamma_E$, the average decoding error probability can be written as
	\begin{align}\label{q18}
	\mathbb{E}_{\gamma _{R},\gamma _{E}} \left [\epsilon\right] =\int_0^\infty f_{\gamma_E}(y)\int_y^\infty \epsilon_{\gamma_R|\gamma_E=y}f_{\gamma_R}(x)\,dx\,dy,
	\end{align}
	for $\gamma_R>\gamma_E$; otherwise $\mathbb{E}_{\gamma _{R},\gamma _{E}} \left [\epsilon\right] =1$ for $\gamma_R\leq\gamma_E$.
	
	The integral in \eqref{q18} does not admit a closed-form solution. To circumvent this difficulty, we employ the linear approximation of $\epsilon$ as follows \cite{Security5}.
	\begin{align} \label{q19}
	\epsilon_{\gamma_R|\gamma_E} \approx \hat{\epsilon}\triangleq \! \begin{cases} 1; 
	&\quad \gamma _{R}\leq v \\[2pt]
	 \displaystyle \frac {1}{2}-\omega(\gamma _{R}-\beta ); &\quad v < \gamma_R < u \\[2mm] 0; 
	 &\quad \gamma _{R}\geq u, \end{cases}
	 \end{align}
	where 
	\begin{align} \label{q20}
	\beta=(1+\gamma_{E})\exp \left ({\frac {\mu V^{\frac {1}{2}}(\gamma_{E})}{\sqrt {m}}\ln 2+\frac {M}{m}\ln 2}\right )-1.
	\end{align}
	and $\omega=\sqrt{m/(2\pi\beta(\beta+2))}$, $v=\beta-\frac{1}{2\omega}$, and $u=\beta+\frac{1}{2\omega}$.
	
	We first derive the conditional expectation $\mathbb{E}[\hat{\epsilon}|\gamma_E = y]$, which is given by
	\begin{align}
	\mathbb{E}\left [\hat{\epsilon}|\gamma_E = y\right] =\int_y^\infty \hat{\epsilon}_{\gamma_R|\gamma_E=y}f_{\gamma_R}(x)\,dx.
	\end{align}
	By comparing the value of $y$ with $v$ and $u$, three cases arise.
	
	\textit{Case 1:} When $y\leq v$, substituting the piecewise approximation \eqref{q19} into the conditional expectation, we obtain
	\begin{align}\label{eq:case1_raw}
	\mathbb{E}\left [\hat{\epsilon}\right] =\underbrace{\int_y^v 1\cdot f_{\gamma_R}(x)\,dx}_{I_1}+\underbrace{\int_v^u \left(\frac {1}{2}-\omega(x-\beta )\right)f_{\gamma_R}(x)\,dx}_{I_2}.
	\end{align}
	The first integral is straightforward:
	\begin{align}\label{eq:I1}
	I_1 = F_{\gamma_R}(v) - F_{\gamma_R}(y).
	\end{align}
	For $I_2$, we define $g(x) \triangleq \frac{1}{2} - \omega(x - \beta)$ and apply integration by parts, yielding
	\begin{align}\label{eq:I2_ibp}
	I_2 = \left[g(x) F_{\gamma_R}(x)\right]_v^u + \omega \int_v^u F_{\gamma_R}(x)\,dx,
	\end{align}
	where we have used $g'(x) = -\omega$. The boundary terms are evaluated as
	\begin{align}
	g(u) &= \tfrac{1}{2} - \omega\!\left(\beta + \tfrac{1}{2\omega} - \beta\right) = 0,\nonumber\\
	g(v) &= \tfrac{1}{2} - \omega\!\left(\beta - \tfrac{1}{2\omega} - \beta\right) = 1.
	\end{align}
	Substituting the boundary values into \eqref{eq:I2_ibp} gives
	\begin{align}\label{eq:I2_result}
	I_2 = -F_{\gamma_R}(v) + \omega \int_v^u F_{\gamma_R}(x)\,dx.
	\end{align}
	Combining \eqref{eq:I1} and \eqref{eq:I2_result}, the $F_{\gamma_R}(v)$ terms cancel, and we arrive at
	\begin{align}\label{eq:case1_result}
	\mathbb{E}\left [\hat{\epsilon}\right] =\omega\int_v^u F_{\gamma_R}(x)\,dx-F_{\gamma_R}(y).
	\end{align}
	To evaluate the integral in \eqref{eq:case1_result}, we apply the change of variables $x = \frac{u-v}{2}t + \frac{u+v}{2}$, which maps $x \in [v, u]$ to $t \in [-1, 1]$. Since $u - v = \frac{1}{\omega}$, we have
	\begin{align}
	\omega\int_v^u F_{\gamma_R}(x)\,dx = \omega \cdot \frac{u-v}{2}\int_{-1}^{1} F_{\gamma_R}(x(t))\,dt = \frac{1}{2}\int_{-1}^{1} F_{\gamma_R}(x(t))\,dt.
	\end{align}
	Applying the Gauss-Chebyshev quadrature \cite{NumericalAnalysis}, i.e., $\int_{-1}^{1} h(t)\,dt \approx \frac{\pi}{U_p}\sum_{p=1}^{U_p}\sqrt{1-t_p^2}\, h(t_p)$, yields
	\begin{align}
	\mathbb{E}\left [\hat{\epsilon}\right] =\frac{\pi}{2 U_p}\sum_{p=1}^{U_p}\sqrt{1-t_p^2}F_{\gamma_R}(x_p)-F_{\gamma_R}(y),
	\end{align}
	where $t_p=\cos\left(\frac{2p-1}{2U_p}\pi\right)$ and $x_p=\frac{1}{2\omega}t_p+\beta$.

	\textit{Case 2:} When $v<y< u$, the conditional expectation becomes
	\begin{align}
	\mathbb{E}\left [\hat{\epsilon}\right] &=\int_y^u  \left(\frac {1}{2}-\omega(x-\beta )\right)f_{\gamma_R}(x)\,dx\nonumber\\
	&=\omega\int_y^u F_{\gamma_R}(x)\,dx-F_{\gamma_R}(v).
	\end{align}
	Applying the Gauss-Chebyshev quadrature yields
	\begin{align}
	\mathbb{E}\left [\hat{\epsilon}\right] =\frac{\pi}{2 U_p}\sum_{p=1}^{U_p}\sqrt{1-t_p^2}F_{\gamma_R}(x_p)-F_{\gamma_R}(v).
	\end{align}

	\textit{Case 3:} When $y\geq u$, the secrecy condition is always satisfied and thus
	\begin{align}
	\mathbb{E}\left [\hat{\epsilon}\right] =0.
	\end{align}
	
	Having obtained $\mathbb{E}[\hat{\epsilon}]$, the overall average decoding error probability is computed by averaging over $\gamma_E$:
	\begin{align}\label{q29}
	\mathbb{E}_{\gamma _{R},\gamma _{E}} \left [\epsilon\right] =\int_0^\infty \mathbb{E}\left [\hat{\epsilon}\right] f_{\gamma_E}(y)\, dy.
	\end{align}
	Since $\mathbb{E}[\hat{\epsilon}]$ is a piecewise function with three cases, the integration interval $[0,\infty)$ in \eqref{q29} must be carefully partitioned. From \eqref{q20}, $\beta$ is a function of $\gamma_E$. The four critical points of $\gamma_E$ that delineate the three cases are determined by solving the following two equations:
	\begin{align}
	\gamma_E=&v,\label{q30}\\
	\gamma_E=&u.\label{q31}
	\end{align}
	Substituting $\omega=\sqrt{m/(2\pi\beta(\beta+2))}$ and $v=\beta-\frac{1}{2\omega}$ into \eqref{q30}, we obtain
	\begin{align}\label{q32}
	(2m-\pi)\beta^2-(4m\gamma_E-2\pi)\beta+2m\gamma_E^2=0.
	\end{align}
	Since $\beta$ is a monotonically increasing function of $\gamma_E$ \cite{LaiX}, the two solutions of $\gamma_E$ to \eqref{q32}, denoted as $\tau_1$ and $\tau_2$ with $\tau_1\leq \tau_2$, can be obtained via the bisection method.
	
	Similarly, substituting the expressions for $\omega$ and $u=\beta+\frac{1}{2\omega}$ into \eqref{q31} yields
	\begin{align}\label{q33}
	(2m-\pi)\beta^2+(4m\gamma_E-2\pi)\beta+2m\gamma_E^2=0.
	\end{align}
	The two solutions of $\gamma_E$ to \eqref{q33} are denoted as $\tau_3$ and $\tau_4$ with $\tau_3\leq \tau_4$, which can also be found via the bisection method.

	Based on these critical points, the integration interval in \eqref{q29} for the three cases is determined as follows.
	
	\textit{Case 1:} When $y\leq v$, i.e., $y\leq \beta-\frac{1}{2\omega}$, the integration interval is
	\begin{align}
	\Theta_1=[\tau_1,\tau_2].
	\end{align}
	
	\textit{Case 2:} When $v<y< u$, i.e., $\beta-\frac{1}{2\omega}<y<\beta+\frac{1}{2\omega}$, the integration interval is
	\begin{align}
	\Theta_2=\left(\left((-\infty,\tau_1]\cup [\tau_2, \infty)\right)\cap [\tau_3,\tau_4]\right).
	\end{align}
	
	\textit{Case 3:} When $y\geq u$, the contribution to $\mathbb{E}_{\gamma _{R},\gamma _{E}} [\epsilon]$ is zero.

	Substituting the integration intervals and the corresponding $\mathbb{E}[\hat{\epsilon}]$ for each case into \eqref{q29}, we apply the Gauss-Chebyshev quadrature to evaluate the outer integral.
	
	For the $i$-th case, $i\in\{1,2\}$, let $\mathcal{H}_i$ and $\mathcal{L}_i$ denote the upper and lower limits of the integration interval, respectively, and let $\mathcal{Y}_{i}(y)$ denote $\mathbb{E}[\hat{\epsilon}]$ as a function of $\gamma_E$ in the $i$-th case. Then, $\mathbb{E}_{\gamma _{R},\gamma _{E}} [\epsilon]$ is computed as
	\begin{align}
	\mathbb{E}_{\gamma _{R},\gamma _{E}} \left [\epsilon\right]_i =\frac{\pi}{U_l}\frac{\mathcal{H}_i-\mathcal{L}_i}{2}\sum_{l=1}^{U_l}\sqrt{1-t_l^2}\mathcal{Y}_{i}(y_l)f_{\gamma_E}(y_l),
	\end{align}
	where $U_l$ is the number of quadrature points controlling the complexity-accuracy tradeoff, and
	\begin{align} 
	t_{l}&=\cos \left ({\frac {(2l-1)\pi }{2U_l}}\right ),\\ 
	y_{l}&=\frac {\mathcal{H}_i-\mathcal{L}_i}{2}t_{l}+\frac {\mathcal{H}_i+\mathcal{L}_i}{2}.
	\end{align}
	
	Accordingly, the AAST is given by
	\begin{align}
	T=
	\begin{cases} \frac{m}{M}\left(1-\mathbb{E}_{\gamma_{R},\gamma_{E}}[\epsilon]_i\right), 
	&\quad i\in\{1,2\}, \\[2pt]
	 \frac{m}{M}, &\quad i=3.
	  \end{cases}
	\end{align}
	
	\begin{remark}[Blocklength-Rate-Security Trade-off]\label{remark:blocklength}
	The AAST expression $T = \frac{m}{M}(1-\mathbb{E}[\epsilon])$ reveals a fundamental three-way trade-off among the transmission rate $m/M$, the decoding reliability $\mathbb{E}[\epsilon]$, and the security level $\delta$. Specifically: \\
	\emph{(i) Blocklength $M$:} Increasing $M$ reduces the decoding error probability $\epsilon$ (since the channel dispersion penalty $\sqrt{V_R/M}$ diminishes), but simultaneously decreases the effective rate $m/M$. This creates a unimodal behavior of $T$ with respect to $M$, confirming the existence of an optimal blocklength $M^*$.\\
	\emph{(ii) Information bits $m$:} For a fixed $M$, increasing $m$ raises the rate $m/M$ but also increases $\epsilon$ (as decoding becomes harder), leading to a similar trade-off and an optimal $m^*$.\\
	\emph{(iii) Secrecy leakage $\delta$:} A smaller $\delta$ enforces stricter secrecy, which increases the effective decoding threshold $\beta$ in \eqref{q20} and consequently raises $\epsilon$, reducing the throughput $T$.\\
	These trade-offs underscore the importance of jointly optimizing $M$, $m$, and $P$ in the FAS-aided secure short-packet system, as addressed in Section~IV.
	\end{remark}
	
	\subsection{Asymptotic AAST}
	To reduce computational complexity, we derive an asymptotic approximation of the AAST. Following \cite{WangH}, we relax the lower limit of the inner integral in \eqref{q18} from $y$ to $0$, yielding
	\begin{align}\label{q42}
	\mathbb{E}_{\gamma _{R},\gamma _{E}} \left [\epsilon\right] \approx\int_0^\infty f_{\gamma_E}(y)\int_0^\infty \epsilon_{\gamma_R|\gamma_E=y}f_{\gamma_R}(x)\,dx\,dy.
	\end{align}
	Substituting the linear approximation \eqref{q19} with the lower limit changed to $0$, we have
	\begin{align}\label{eq:asym_raw}
	\mathbb{E}\left [\hat{\epsilon}\right]=\underbrace{\int_0^v f_{\gamma_R}(x)\,dx}_{J_1}+\underbrace{\int_v^u \left(\frac {1}{2}-\omega(x-\beta )\right)f_{\gamma_R}(x)\,dx}_{J_2}.
	\end{align}
	The first integral evaluates to $J_1 = F_{\gamma_R}(v) - F_{\gamma_R}(0) = F_{\gamma_R}(v)$, since $F_{\gamma_R}(0) = 0$ for a non-negative random variable. For $J_2$, applying integration by parts as in \eqref{eq:I2_ibp}--\eqref{eq:I2_result} yields $J_2 = -F_{\gamma_R}(v) + \omega \int_v^u F_{\gamma_R}(x)\,dx$. Combining $J_1$ and $J_2$, the $F_{\gamma_R}(v)$ terms cancel, giving
	\begin{align}\label{eq:asym_simplified}
	\mathbb{E}\left [\hat{\epsilon}\right]=\omega\int_v^u F_{\gamma_R}(x)\,dx.
	\end{align}
	Applying the same Gauss-Chebyshev quadrature procedure as in \eqref{eq:case1_result} yields
	\begin{align}
	\mathbb{E}\left [\hat{\epsilon}\right]=\frac{\pi}{2 U_p}\sum_{p=1}^{U_p}\sqrt{1-t_p^2}F_{\gamma_R}(x_p).
	\end{align}
	Similarly, applying the Gauss-Chebyshev quadrature to the outer integral over $y$ gives
	\begin{align}
	\mathbb{E}_{\gamma _{R},\gamma _{E}} \left [\epsilon\right] =\frac{\mathcal{H}}{2}\frac{\pi}{U_l}\sum_{l=1}^{U_l}\sqrt{1-t_l^2}f_{\gamma_E}(y_l)\mathbb{E}\left [\hat{\epsilon}\right],
	\end{align}
	where $\mathcal{H}$ is a sufficiently large truncation parameter for the integration range. Consequently, the asymptotic AAST is obtained as
	\begin{align}\label{q46}
	\hat{T}= \frac{m}{M}\left(1-\frac{\mathcal{H}}{2}\frac{\pi}{U_l}\sum_{l=1}^{U_l}\sqrt{1-t_l^2}f_{\gamma_E}(y_l)\mathbb{E}\left [\hat{\epsilon}\right]\right).
	\end{align}
	
	\begin{remark}[High-SNR Behavior and Design Guidelines]\label{remark:high_snr}
	The asymptotic AAST in \eqref{q46} provides several important design insights for FAS-aided secure short-packet systems.\\
	\emph{(i) High-SNR regime:} As $\bar{\gamma}_R \to \infty$, the CDF $F_{\gamma_R}(x_p) = F_{\alpha_R}\left(\sqrt{x_p/\bar{\gamma}_R}\right) \to 0$ for any finite $x_p$, which implies $\mathbb{E}[\hat{\epsilon}] \to 0$ from \eqref{eq:asym_simplified}. Consequently, $\hat{T} \to m/M$, indicating that the AAST converges to the nominal transmission rate in the high-SNR regime. This confirms that adequate transmit power can effectively overcome the finite blocklength penalty.\\
	\emph{(ii) Power-distance scaling:} From the average SNR definitions $\bar{\gamma}_R = P d_R^{-a_R}/\sigma_R^2$ and $\bar{\gamma}_E = P d_E^{-a_E}/\sigma_E^2$, the ratio $\bar{\gamma}_R/\bar{\gamma}_E = d_E^{a_E}/d_R^{a_R} \cdot \sigma_E^2/\sigma_R^2$ is independent of $P$. This implies that increasing the transmit power alone improves both the legitimate and eavesdropping SNRs proportionally, and the secrecy advantage is fundamentally determined by the path-loss ratio $d_E^{a_E}/d_R^{a_R}$. Therefore, \emph{FAS port selection} plays a critical role in creating an additional secrecy advantage beyond path loss, as it independently enhances the legitimate channel quality through spatial diversity.\\
	\emph{(iii) Practical design guideline:} For a target secrecy throughput $T_{\mathrm{target}}$, the system designer should first select $N_R$ to ensure sufficient diversity gain over the eavesdropper (cf.\ Remark~\ref{remark:diversity}), then jointly optimize $P$ and $M$ to balance the rate-reliability trade-off (cf.\ Remark~\ref{remark:blocklength}).
	\end{remark}
	
	\section{Optimization Problem Formulation and Algorithm Design}
	
	In this section, we formulate the optimization problem with the objective of maximizing the asymptotic AAST in \eqref{q46} under practical constraints, and present a grid search (GS) algorithm to find the optimal operating points for the FAS system.
	
	\subsection{FAS Security Optimization Problem}
	
	The objective is to maximize the asymptotic AAST while satisfying practical constraints on transmit power, port number, and blocklength. This leads to the following optimization problem:
	\begin{align}\label{eq:optimization_problem}
		\max_{N_R, P, M} \quad &\hat{T}(N_R, P, M) \\
		\text{s.t.  } \quad &P \leq P_{\max}, \nonumber \\
		&N_R \leq N_{\max}, \nonumber \\
		&M \leq M_{\max}, \nonumber \\
		&\hat{T}(N_R, P, M) \geq 0, \nonumber
	\end{align}
	where $\hat{T}(N_R, P, M)$ denotes the asymptotic AAST as a function of the RU's port number $N_R$, transmit power $P$, and blocklength $M$, as derived in \eqref{q46}.
	
	The optimization problem in \eqref{eq:optimization_problem} is inherently complex due to the non-convex nature of the objective function and the discrete nature of the port selection variable $N_R$. Moreover, the presence of the eavesdropper with similar FAS capabilities creates a competitive scenario where the optimization landscape exhibits multiple local optima.
	
	To address this optimization problem, we first analyze its structural properties and then propose the GS algorithm.
	
	\begin{proposition}[Monotonicity of $\hat{T}$ with respect to $N_R$]\label{prop:monotone}
		For fixed $P$ and $M$, the asymptotic AAST $\hat{T}(N_R, P, M)$ in \eqref{q46} is monotonically non-decreasing in $N_R$.
	\end{proposition}
	
	\begin{proof}
		From \eqref{eq:cdf_final}, the CDF of the maximum channel amplitude at the RU is $F_{\alpha_R}(x) = \prod_{d=1}^{D} F_{\alpha_d}(x)$. When $N_R$ increases, the VBCM framework either increases the number of blocks $D$ or the block sizes $\{L_d\}$. In both cases, since each block CDF satisfies $0 \leq F_{\alpha_d}(x) \leq 1$ for $x \geq 0$, the product $F_{\alpha_R}(x)$ is non-increasing in $N_R$ for any given $x > 0$. Consequently, from \eqref{eq:cdf_snr_R}, $F_{\gamma_R}(z)$ is also non-increasing in $N_R$ for any $z > 0$. From \eqref{eq:asym_simplified}, the conditional expectation $\mathbb{E}[\hat{\epsilon}] = \omega \int_v^u F_{\gamma_R}(x)\,dx$ is non-increasing in $N_R$ since the integrand decreases. Since $\hat{T} = \frac{m}{M}\left(1 - \mathbb{E}_{\gamma_R, \gamma_E}[\epsilon]\right)$ and $\mathbb{E}_{\gamma_R, \gamma_E}[\epsilon]$ is non-increasing in $N_R$, the asymptotic AAST $\hat{T}$ is monotonically non-decreasing in $N_R$.
	\end{proof}
	
	\begin{corollary}[Dimensionality Reduction]\label{cor:dim_reduction}
		Under the optimization problem \eqref{eq:optimization_problem}, the optimal number of RU ports is $N_R^* = N_{\max}$. Therefore, the three-dimensional search over $(N_R, P, M)$ reduces to a two-dimensional search over $(P, M)$:
		\begin{align}\label{eq:reduced_problem}
			\max_{P, M} \quad &\hat{T}(N_{\max}, P, M), \quad \text{s.t. } P \leq P_{\max},\ M \leq M_{\max}.
		\end{align}
		The total number of candidate solutions reduces from $|\mathcal{S}|$ in \eqref{eq:total_candidates} to
		\begin{align}\label{eq:reduced_candidates}
			|\mathcal{S}'| = G_P \times G_M,
		\end{align}
		yielding a computational saving factor of $(N_R^{\max} - N_R^{\min} + 1)$.
	\end{corollary}
	
	\begin{remark}[Problem Structure]\label{remark:structure}
		The optimization problem \eqref{eq:optimization_problem} is a mixed-integer non-convex program, since $N_R$ is discrete and $\hat{T}$ is generally non-convex in $(P, M)$ due to the composite structure involving products of CDFs and Gauss-Chebyshev sums. However, Proposition~\ref{prop:monotone} resolves the integer component, and the reduced problem \eqref{eq:reduced_problem} involves only two continuous variables, making the GS algorithm particularly efficient.
	\end{remark}
	
	\subsection{Grid Search Optimization}
	
	The GS algorithm provides a systematic approach to explore the entire feasible parameter space, guaranteeing global optimality at the cost of computational complexity.
	
	Specifically, the GS method discretizes problem \eqref{eq:optimization_problem} into a finite set of candidate solutions. For the transmit power $P$ and the blocklength $M$, we construct uniform grids as
	\begin{align}
		\mathcal{P} = &\left\{P^{\min} + j \cdot \frac{P^{\max} - P^{\min}}{G_P-1} : j = 0, 1, \ldots, G_P-1\right\},\label{eq:power_grid}\\
		\mathcal{M} = &\left\{M^{\min} + t \cdot \frac{M^{\max} - M^{\min}}{G_M-1} : t = 0, 1, \ldots, G_M-1\right\},\label{eq:BL_grid}
	\end{align}
	where $G_P$ and $G_M$ denote the grid resolutions for $P$ and $M$, respectively. The corresponding grid spacings are
	\begin{align}\label{eq:grid_spacing}
		\Delta P = \frac{P^{\max} - P^{\min}}{G_P-1},\quad
		\Delta M = \frac{M^{\max} - M^{\min}}{G_M-1}.
	\end{align} 
	For the discrete port parameter $N_R$, the search space is naturally defined as
	\begin{align}\label{eq:port_grid}
		\mathcal{N} = \{N_R^{\min}, N_R^{\min}+1, \ldots, N_R^{\max}\}.
	\end{align}
	Accordingly, the total number of candidate solutions is
	\begin{align}\label{eq:total_candidates}
		|\mathcal{S}| = |\mathcal{P}| \times |\mathcal{M}|\times |\mathcal{N}| = G_P \times G_M \times (N_R^{\max} - N_R^{\min} + 1).
	\end{align}
	For each candidate $(N_R, P, M) \in \mathcal{N} \times \mathcal{P} \times \mathcal{M}$, the asymptotic AAST is computed using \eqref{q46}. The GS algorithm seeks the global optimum, which is expressed as
	\begin{align}\label{eq:grid_optimization}
		(N_R^*, P^*, M^*) = \arg\max_{(N_R, P, M) \in \mathcal{N} \times \mathcal{P}\times \mathcal{M}} \hat{T}(N_R, P, M).
	\end{align}
	Equivalently, the optimization can be decomposed as
	\begin{align}\label{eq:grid_process}
		\hat{T}^* &= \max_{N_R \in \mathcal{N}} \max_{P \in \mathcal{P}}\max_{M \in \mathcal{M}} \hat{T}(N_R, P, M) \nonumber\\
		&= \max_{i=1,\ldots,|\mathcal{N}|} \max_{j=1,\ldots,G_P} \max_{t=1,\ldots,G_M}\hat{T}(N_R^{(i)}, P^{(j)}, M^{(t)}),
	\end{align}
	where $N_R^{(i)}$, $P^{(j)}$, and $M^{(t)}$ denote the $i$-th, $j$-th, and $t$-th grid points, respectively. The overall GS algorithm is summarized in Algorithm~\ref{alg:grid_search}.
	
	\begin{algorithm}[h!]
		\caption{Grid Search Optimization for FAS SP Security}
		\label{alg:grid_search}
		\begin{algorithmic}[1]
			\Require $N_E$, $W$, $\eta_R$, $\eta_E$, $\sigma_R^2$, $\sigma_E^2$, $P^{\min}$, $P^{\max}$, $N_R^{\min}$, $N_R^{\max}$, $M^{\min}$, $M^{\max}$, $G_P$, $G_M$
			\Ensure Optimal $(N_R^*, P^*, M^*)$ and maximum asymptotic AAST $\hat{T}^*$
			\State Initialize $\hat{T}^* \leftarrow 0$, $N_R^* \leftarrow N_R^{\min}$, $P^* \leftarrow P^{\min}$, $M^* \leftarrow M^{\min}$
			\State Construct power grid $\mathcal{P}$ using \eqref{eq:power_grid}
			\State Construct blocklength grid $\mathcal{M}$ using \eqref{eq:BL_grid}
			\State Construct port grid $\mathcal{N}$ using \eqref{eq:port_grid}
			\For{$N_R \in \mathcal{N}$}
			\For{$P \in \mathcal{P}$}
			\For{$M \in \mathcal{M}$}
			\State Compute $\hat{T}(N_R, P, M)$ using \eqref{q46}
			\If{$\hat{T}(N_R, P, M)> \hat{T}^*$}
			\State $\hat{T}^* \leftarrow \hat{T}(N_R, P, M)$
			\State $N_R^* \leftarrow N_R$, $P^* \leftarrow P$, $M^* \leftarrow M$
			\EndIf
			\EndFor
			\EndFor
			\EndFor
			\State \Return $(N_R^*, P^*, M^*, \hat{T}^*)$
		\end{algorithmic}
	\end{algorithm}

	The GS algorithm achieves global optimality in the discrete sense. As the grid resolution increases, the approximation error decreases. By the Lipschitz continuity of $\hat{T}$ with respect to both $P$ and $M$, the overall grid approximation error is bounded by
	\begin{align}\label{eq:grid_error}
		\epsilon_{\mathrm{grid}} = \left|\hat{T}^{\mathrm{true}} - \hat{T}^*\right| \leq L_P \cdot \frac{\Delta P}{2} + L_M \cdot \frac{\Delta M}{2},
	\end{align}
	where $L_P$ and $L_M$ denote the Lipschitz constants of $\hat{T}$ with respect to $P$ and $M$, respectively, and $\hat{T}^{\mathrm{true}}$ is the true continuous optimum. Substituting \eqref{eq:grid_spacing} into \eqref{eq:grid_error}, the error can be expressed in terms of the grid resolutions as
	\begin{align}\label{eq:grid_error_explicit}
		\epsilon_{\mathrm{grid}} \leq \frac{L_P(P^{\max} - P^{\min})}{2(G_P - 1)} + \frac{L_M(M^{\max} - M^{\min})}{2(G_M - 1)}.
	\end{align}
	For a given target accuracy $\epsilon_{\mathrm{target}} > 0$, by allocating equal error budgets to each dimension, the minimum grid resolutions are
	\begin{align}\label{eq:min_grid_P}
		G_P^{\min} &= 1 + \left\lceil \frac{L_P(P^{\max} - P^{\min})}{\epsilon_{\mathrm{target}}} \right\rceil,\\
		G_M^{\min} &= 1 + \left\lceil \frac{L_M(M^{\max} - M^{\min})}{\epsilon_{\mathrm{target}}} \right\rceil, \label{eq:min_grid_M}
	\end{align}
	where $\lceil \cdot \rceil$ denotes the ceiling function. Equations \eqref{eq:min_grid_P}--\eqref{eq:min_grid_M} provide a principled way to select the grid resolution based on the desired accuracy.
	
	The computational complexity of the GS algorithm is $\mathcal{O}(G_P\cdot G_M \cdot (N_R^{\max} - N_R^{\min} + 1))$, where each objective function evaluation requires numerical integration with complexity $\mathcal{O}(U_p \cdot U_l)$. By leveraging Corollary~\ref{cor:dim_reduction}, the reduced complexity becomes $\mathcal{O}(G_P \cdot G_M)$. The total computational cost is
	\begin{align}\label{eq:grid_complexity}
		\mathcal{T}_{\mathrm{grid}} = G_P\cdot G_M \cdot U_p \cdot U_l \cdot \mathcal{T}_{\mathrm{eval}},
	\end{align}
	where $\mathcal{T}_{\mathrm{eval}}$ denotes the cost of evaluating the VBCM channel statistics. Substituting \eqref{eq:min_grid_P}--\eqref{eq:min_grid_M}, the minimum computational cost to achieve accuracy $\epsilon_{\mathrm{target}}$ scales as
	\begin{align}\label{eq:min_cost}
		\mathcal{T}_{\mathrm{grid}}^{\min} = \mathcal{O}\!\left(\frac{L_P L_M (P^{\max}\!-\!P^{\min})(M^{\max}\!-\!M^{\min})}{\epsilon_{\mathrm{target}}^2} \cdot U_p U_l \right)\!.
	\end{align}
	This reveals a fundamental accuracy-complexity trade-off: halving the target error requires quadrupling the computational effort.

	\subsection{Theoretical Convergence Analysis}
	
	\begin{theorem}[Grid Search Optimality]
		Algorithm~\ref{alg:grid_search} converges to the global optimum of problem \eqref{eq:optimization_problem} as the grid resolutions $G_P \to \infty$ and $G_M \to \infty$.
	\end{theorem}
	
	\begin{proof}
		Let $\mathcal{S}^* = \arg\max_{(N_R, P, M) \in \mathcal{X}} \hat{T}(N_R, P, M)$ denote the set of global optima, where $\mathcal{X}$ is the feasible region. For any $\epsilon_P > 0$ and $\epsilon_M > 0$, there exist grid resolutions $G_{\epsilon_P}$ and $G_{\epsilon_M}$ such that for all $G_P \geq G_{\epsilon_P}$ and $G_M \geq G_{\epsilon_M}$, the grid spacings satisfy
		\begin{align}
			\Delta P = \frac{P^{\max} - P^{\min}}{G_P-1} \leq \frac{\epsilon_P}{L_P},\quad
			\Delta M = \frac{M^{\max} - M^{\min}}{G_M-1} \leq \frac{\epsilon_M}{L_M},
		\end{align}
		where $L_P$ and $L_M$ are the Lipschitz constants of $\hat{T}$ with respect to $P$ and $M$, respectively.
		
		For any global optimum $(N_R^*, P^*, M^*) \in \mathcal{S}^*$, there exists a grid point $(N_R^{(i)}, P^{(j)}, M^{(t)})$ such that $|P^* - P^{(j)}| \leq \Delta P/2$ and $|M^* - M^{(t)}| \leq \Delta M/2$. By Lipschitz continuity, 
		\begin{align}
			&|\hat{T}(N_R^*, P^*, M^*) - \hat{T}(N_R^*, P^{(j)}, M^{(t)})|  \nonumber\\
			&\leq  L_P  \cdot \frac{\Delta P}{2} + L_M \cdot \frac{\Delta M}{2}\nonumber\\
			 &\leq  \frac{\epsilon_P + \epsilon_M}{2}.
		\end{align}
		
		Since the GS evaluates all grid points, the maximum over the grid approaches the global maximum as $G_P\to \infty$ and $G_M\to \infty$. Therefore,
		\begin{align}
			\lim_{G_P \to \infty, G_M\to \infty} \hat{T}^{\mathrm{grid}} = \max_{(N_R, P, M) \in \mathcal{X}} \hat{T}(N_R, P, M).
		\end{align}
	\end{proof}
	
	\begin{remark}[Practical Implementation Guidelines]\label{remark:practical}
		The analytical results from Section~III can be leveraged to improve the practical efficiency of Algorithm~\ref{alg:grid_search} in two ways.\\
		\emph{(i) Search range reduction:} From Remark~\ref{remark:high_snr}, the power-distance scaling law $\bar{\gamma}_R/\bar{\gamma}_E = d_E^{a_E}/d_R^{a_R} \cdot \sigma_E^2/\sigma_R^2$ suggests that a meaningful lower bound on $P$ can be determined by requiring $\bar{\gamma}_R > \bar{\gamma}_E$, i.e., $P > \sigma_R^2 d_R^{a_R} / (\sigma_E^2 d_E^{a_E})$, below which the secrecy throughput is negligible. This narrows the search range $[P^{\min}, P^{\max}]$ and reduces $G_P^{\min}$ in \eqref{eq:min_grid_P}. Similarly, the unimodal behavior of $\hat{T}$ with respect to $M$ (cf.\ Remark~\ref{remark:blocklength}) implies that the search range for $M$ can be centered around the approximate optimal point.\\
		\emph{(ii) Adaptive grid refinement:} A coarse-to-fine strategy can be employed, where an initial coarse grid identifies the promising region, followed by a refined grid in the neighborhood of the coarse optimum. This reduces the effective computational cost from $\mathcal{O}(1/\epsilon_{\mathrm{target}}^2)$ in \eqref{eq:min_cost} to $\mathcal{O}(\log^2(1/\epsilon_{\mathrm{target}}))$ while maintaining the global convergence guarantee of Theorem~1.
	\end{remark}
	
	\section{Numerical Results and Analysis}
	
	In this section, numerical results are provided to validate the accuracy of the derived analytical AAST expressions and to evaluate the performance of the proposed GS optimization algorithm under various system configurations.
	
	The system operates at a carrier frequency of $f_c = 2.4$ GHz with wavelength $\lambda = c/f_c = 0.125$ m. The antenna configuration employs a one-dimensional linear fluid antenna with half-wavelength port spacing $\lambda/2$, utilizing the maximum channel gain port selection strategy. The channel model parameters are set as follows: average channel power gains $\eta_R = \eta_E = 1.0$, noise variances $\sigma_R^2 = \sigma_E^2 = 1.0$, and spatial correlation following the Jakes model. The distances from the BS to the RU and the EU are $d_R = 10$ m and $d_E = 100$ m, respectively, and the path-loss exponents are $a_R = a_E = 2$.
	
	For accurate evaluation of the derived analytical expressions, the numerical integration parameters are set as: truncation parameter $\mathcal{H} = 10\sqrt{\eta_R}$, outer quadrature points $U_p = 100$, and inner quadrature points $U_l = 100$.
	
	The GS algorithm is configured with parameter ranges $N_R \in [5, 20]$, $P \in [0.1, 15]$ W, and $M \in [100, 1000]$, subject to the constraint $N_R + N_E \leq 40$. Monte Carlo simulations with $10^5$ independent channel realizations are performed for validation.

	\begin{figure}[t]  \centering
		\includegraphics[width=3.8in]{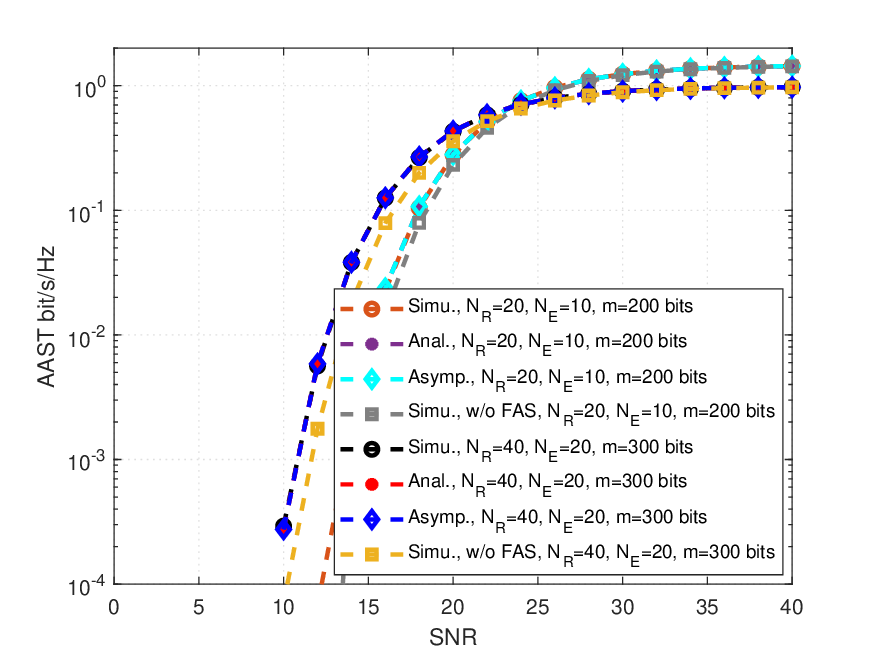}
		\caption {AAST versus SNR, where $m= 200$ bits, $N_R=20$, $N_E=10$, and $m= 300$ bits, $N_R=40$, $N_E=20$, respectively.}\label{AASTvsSNR}
	\end{figure}

	Fig.~\ref{AASTvsSNR} illustrates the AAST as a function of SNR under two system configurations: (i) $m=200$ bits, $N_R=20$, $N_E=10$, and (ii) $m=300$ bits, $N_R=40$, $N_E=20$. Several key observations can be drawn. First, the analytical results (``Anal.'') and the Monte Carlo simulation results (``Simu.'') are in excellent agreement across the entire SNR range for both configurations, thereby validating the accuracy of the derived closed-form AAST expressions. Second, the asymptotic approximation (``Asymp.'') closely tracks the exact analytical curves, particularly in the medium-to-high SNR regime (SNR~$\geq 15$~dB), confirming the tightness of the asymptotic expression in \eqref{q46}. Third, the FAS-aided system achieves a substantial AAST gain over the conventional fixed-position antenna system (``w/o FAS''), with an improvement of approximately one order of magnitude in the mid-SNR region ($15$--$25$~dB). This significant gap underscores the effectiveness of FAS port selection in exploiting spatial diversity to enhance secrecy performance, consistent with the analysis in Remark~\ref{remark:diversity}. Fourth, comparing the two configurations, the system with larger antenna arrays ($N_R=40$, $N_E=20$) achieves a higher AAST at high SNR values despite transmitting more bits ($m=300$), demonstrating that the diversity gain from increased $N_R$ can more than compensate for the higher decoding burden. Finally, the AAST monotonically increases with SNR and converges to a saturation level at high SNR, which is consistent with the high-SNR limiting behavior $\hat{T} \to m/M$ analyzed in Remark~\ref{remark:high_snr}.

	\begin{figure}[t]  \centering
		\includegraphics[width=3.8in]{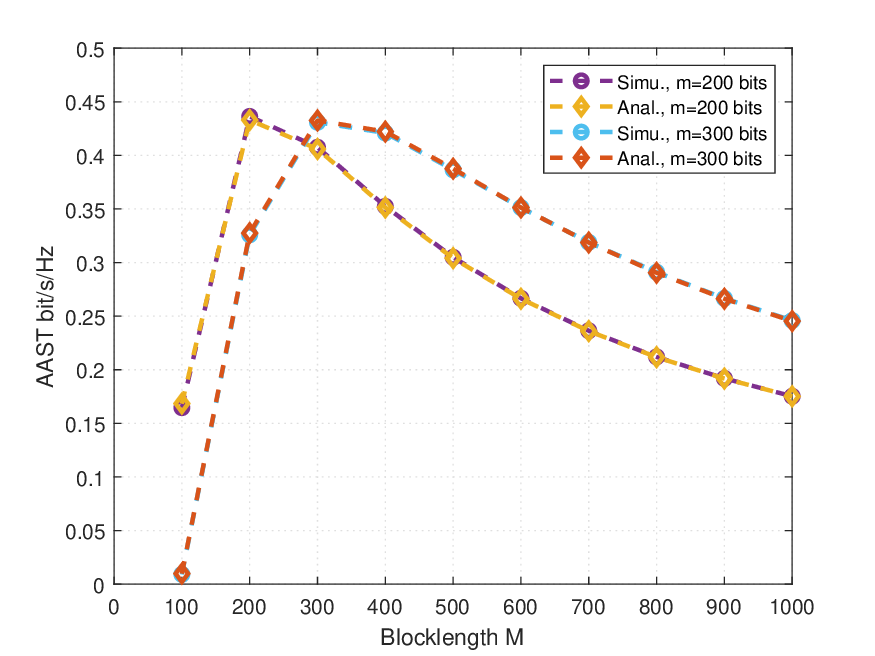}
		\caption { AAST vs Blocklength $M$, where $N_R=40$, $N_E=20$, and SNR =20 dB.}\label{AASTvsM}
	\end{figure}

	Fig.~\ref{AASTvsM} presents the AAST as a function of the blocklength $M$ for $N_R=40$, $N_E=20$, SNR $= 20$~dB, with $m=200$ bits and $m=300$ bits, respectively. The analytical results (``Anal.'') closely match the simulation results (``Simu.'') for both values of $m$, further validating the derived expressions. A distinct unimodal behavior is observed: the AAST first increases rapidly with $M$, reaches a peak, and then gradually decreases. Specifically, for $m=200$ bits, the optimal blocklength is approximately $M^* \approx 200$--$250$, yielding a peak AAST of about $0.44$~bit/s/Hz; for $m=300$ bits, the optimal blocklength shifts to $M^* \approx 250$--$300$, with a comparable peak AAST of about $0.43$~bit/s/Hz. This unimodal behavior confirms the fundamental blocklength-rate trade-off analyzed in Remark~\ref{remark:blocklength}: when $M$ is small, the channel dispersion penalty $\sqrt{V_R/M}$ is large, resulting in a high decoding error probability and low throughput; as $M$ increases beyond the optimal point, the effective transmission rate $m/M$ decreases, which dominates and reduces the throughput. Furthermore, the optimal $M^*$ increases with $m$, since transmitting more information bits requires a longer blocklength to maintain a low decoding error probability. At large blocklengths ($M = 1000$), the $m=300$ configuration retains a higher AAST ($\approx 0.25$~bit/s/Hz) than the $m=200$ case ($\approx 0.17$~bit/s/Hz), because the rate $m/M$ is higher for larger $m$ when $M$ is fixed. These results strongly motivate the inclusion of $M$ as an optimization variable in the formulated problem \eqref{eq:optimization_problem}.

	\begin{figure}[t]  \centering
		\includegraphics[width=3.8in]{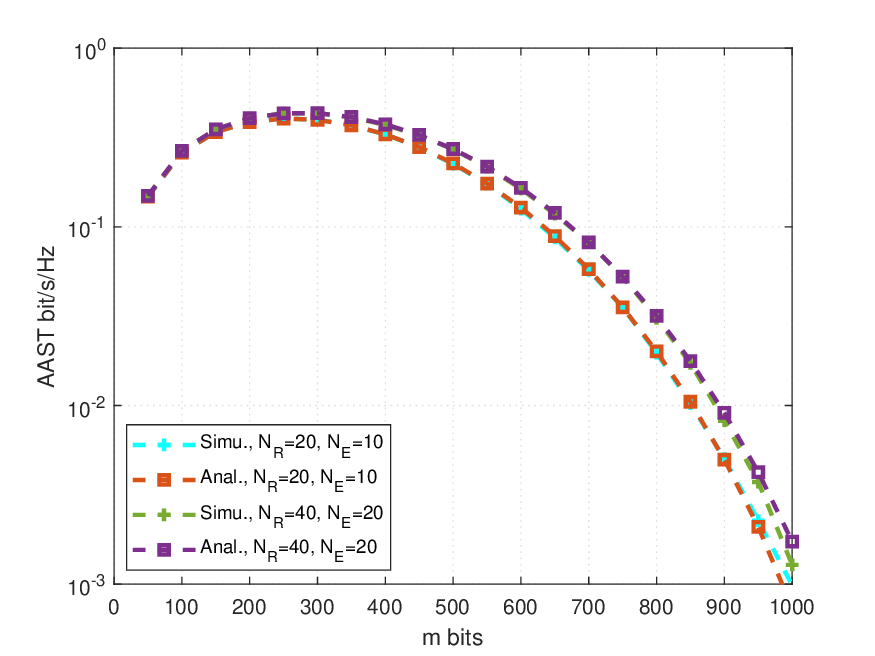}
		\caption { AAST vs transmit bits $m$, where SNR=20 dB, and $N_R=20$, $N_E=10$, and $N_R=40$, $N_E=20$, respectively.}\label{AASTvsm}
	\end{figure}
	
	Fig.~\ref{AASTvsm} investigates the impact of the number of transmit bits $m$ on the AAST for two configurations: $(N_R=20, N_E=10)$ and $(N_R=40, N_E=20)$, with SNR $= 20$~dB. Similar to the blocklength trade-off observed in Fig.~\ref{AASTvsM}, the AAST exhibits a unimodal behavior with respect to $m$: it first increases sharply, reaches a peak, and then decreases as $m$ grows further. For the $(N_R=20, N_E=10)$ configuration, the peak AAST of approximately $0.35$~bit/s/Hz is achieved around $m \approx 300$ bits, while the $(N_R=40, N_E=20)$ configuration attains a slightly higher peak of approximately $0.38$~bit/s/Hz around $m \approx 350$ bits. This confirms the trade-off discussed in Remark~\ref{remark:blocklength}(ii): increasing $m$ raises the effective rate $m/M$ but simultaneously increases the decoding error probability $\epsilon$, and the balance between these two effects determines the optimal $m^*$. Notably, the two configurations yield comparable peak AAST values despite the $(N_R=40, N_E=20)$ case having twice the number of ports, because the eavesdropper's capability also doubles, largely offsetting the RU's diversity gain. Furthermore, when $m$ exceeds approximately $700$ bits, the AAST drops sharply and approaches zero near $m = 1000$ bits, indicating that attempting to transmit too many information bits within a finite blocklength leads to severe decoding failure. The analytical and simulation results remain in close agreement throughout, further validating the derived expressions.

	\begin{figure}[t]  \centering
		\includegraphics[width=3.8in]{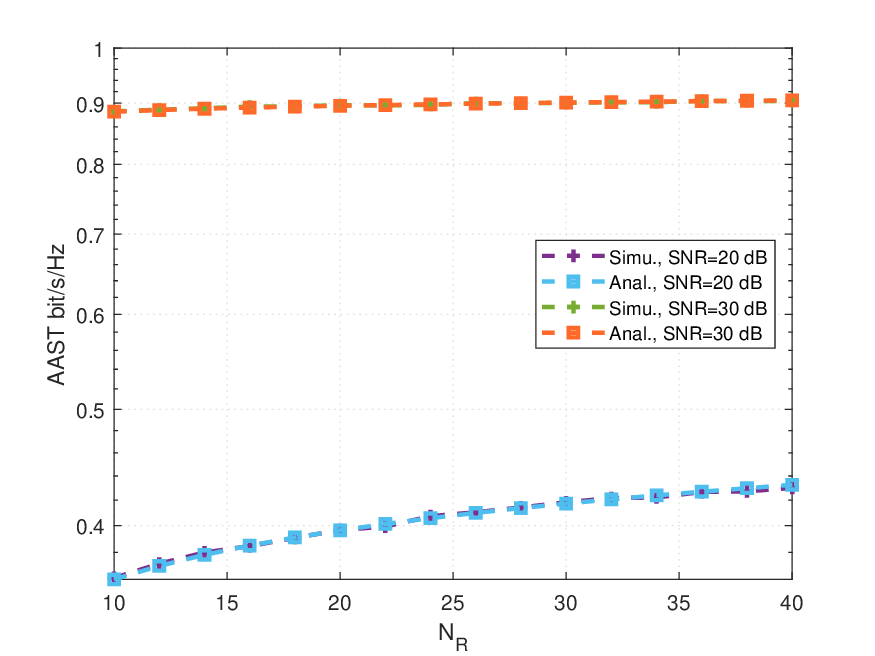}
		\caption { AAST vs  $N_R$, where $m= 300$ bits,  and $N_E=20$.}\label{AASTvsNR}
	\end{figure}
	
	Fig.~\ref{AASTvsNR} depicts the AAST as a function of the number of RU ports $N_R$ for $m = 300$ bits and $N_E = 20$, under two SNR levels: $20$~dB and $30$~dB. Several important observations emerge. First, the AAST is monotonically non-decreasing in $N_R$ for both SNR levels, which is consistent with Proposition~\ref{prop:monotone} and confirms that increasing the number of FAS ports always benefits the secrecy throughput. This also validates Corollary~\ref{cor:dim_reduction}, which states that the optimal port number is $N_R^* = N_{\max}$. Second, the AAST at SNR $= 30$~dB is significantly higher than at SNR $= 20$~dB across all values of $N_R$. Specifically, at SNR $= 30$~dB, the AAST remains above $0.88$~bit/s/Hz and is relatively flat, indicating that the system operates near the saturation regime where $\hat{T} \to m/M$. In contrast, at SNR $= 20$~dB, the AAST increases from approximately $0.38$~bit/s/Hz at $N_R = 10$ to about $0.45$~bit/s/Hz at $N_R = 40$, demonstrating that port diversity provides a more pronounced improvement in the moderate-SNR regime where the system has not yet saturated. Third, the analytical and simulation results are in excellent agreement for both SNR levels, further corroborating the accuracy of the derived AAST expressions.

	\begin{remark}[Asymmetric Port Advantage and Operating Regime]\label{remark:port_regime}
		The numerical results from Figs.~\ref{AASTvsSNR}--\ref{AASTvsNR} reveal two complementary design principles for FAS-aided secure short-packet systems. First, the secrecy gain of FAS is determined by the \emph{port count asymmetry} $N_R - N_E$, rather than the absolute port numbers. This is evidenced by Fig.~\ref{AASTvsm}, where proportionally scaling both $N_R$ and $N_E$ (from $(20, 10)$ to $(40, 20)$) yields only a marginal AAST improvement, whereas Fig.~\ref{AASTvsNR} shows that increasing $N_R$ alone (with $N_E$ fixed) produces a consistent and significant gain. This observation aligns with the analytical insight in Remark~\ref{remark:diversity}: the secrecy throughput is governed by the channel asymmetry between the RU and the EU, which is enhanced by a larger $N_R/N_E$ ratio. Second, from Figs.~\ref{AASTvsSNR} and \ref{AASTvsNR}, the benefit of FAS port diversity is most pronounced in the \emph{moderate-SNR regime} (e.g., $15$--$25$~dB), where both the FAS-versus-FPA gap (Fig.~\ref{AASTvsSNR}) and the AAST slope with respect to $N_R$ (Fig.~\ref{AASTvsNR}) are maximized. At high SNR, the system saturates at $\hat{T} \to m/M$ regardless of $N_R$, and the diversity gain diminishes. These findings suggest that in practical deployments, system designers should maximize the port count ratio $N_R/N_E$ and allocate FAS resources preferentially to moderate-SNR users where the secrecy gain per additional port is highest.
	\end{remark}
	
\begin{figure*}[t]  \centering
		\includegraphics[width=7in]{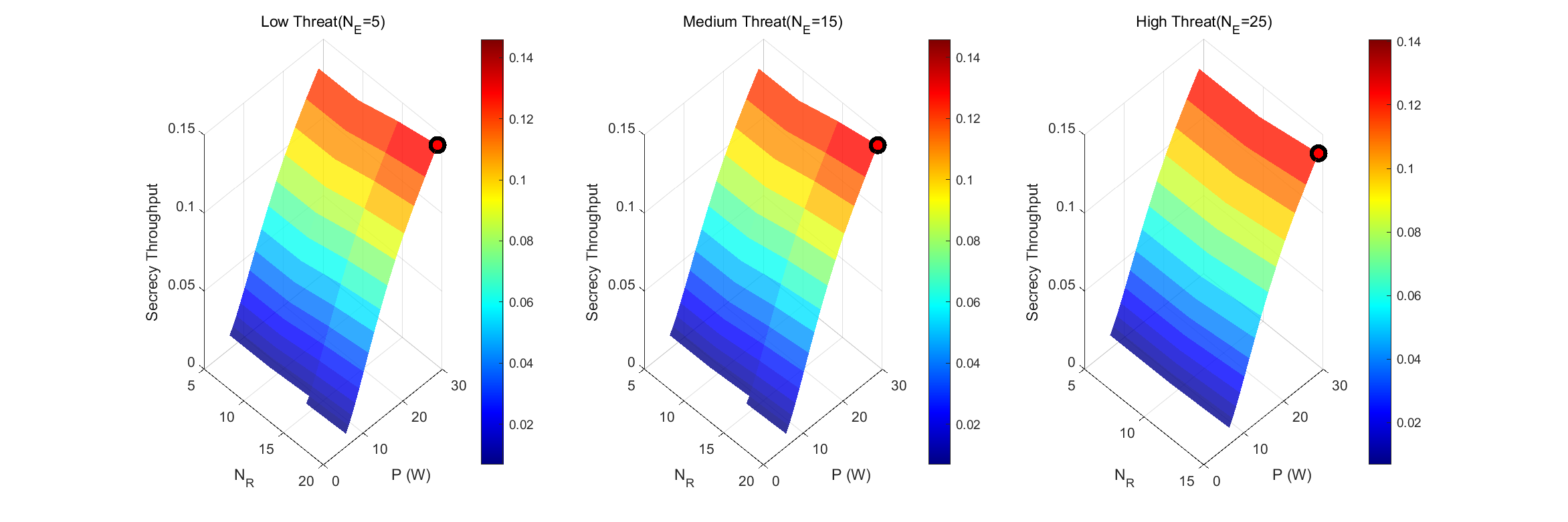}
		\caption{FAS Security Optimization Under Varying Threat Scenarios: 3D Secrecy Throughput Surface Analysis vs $P$ and $N_R$ for Different Eavesdropper Capabilities ($N_E = 5, 15, 25$).}\label{threat_ana1}
	\end{figure*}

	\begin{figure*}[t]  \centering
		\includegraphics[width=7in]{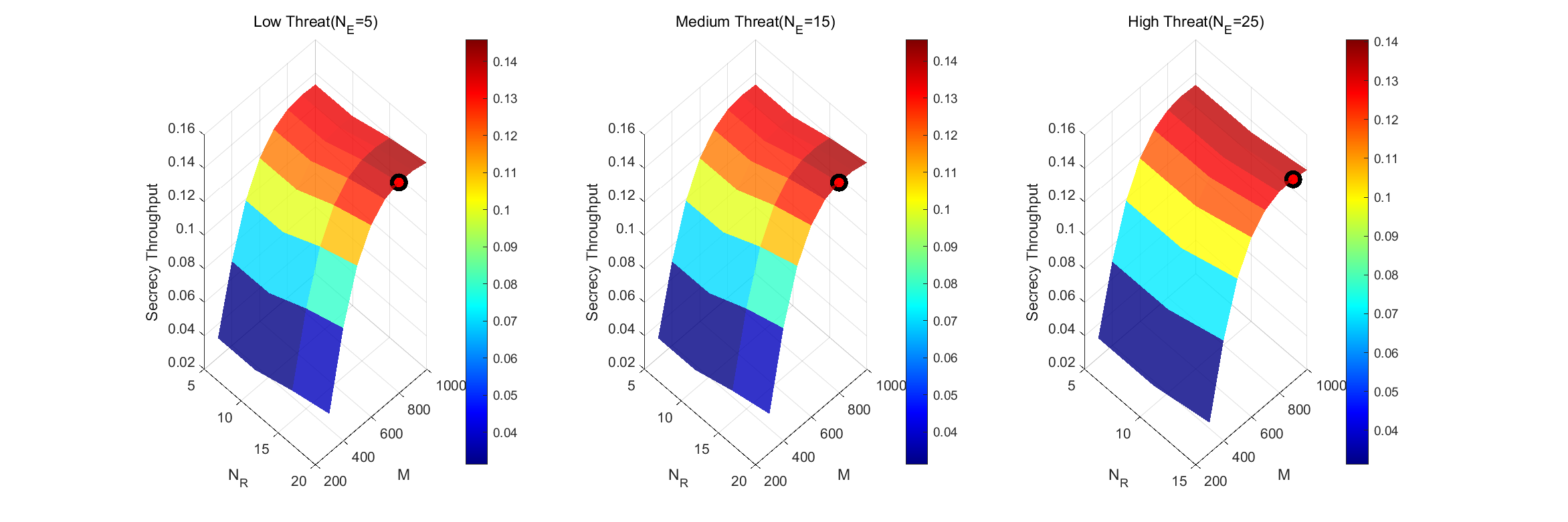}
		\caption{FAS Security Optimization Under Varying Threat Scenarios: 3D Secrecy Throughput Surface Analysis vs $M$ and $N_R$ for Different Eavesdropper Capabilities ($N_E = 5, 15, 25$).}\label{threat_ana2}
	\end{figure*}
	
	\begin{figure*}[t]  \centering
		\includegraphics[width=7in]{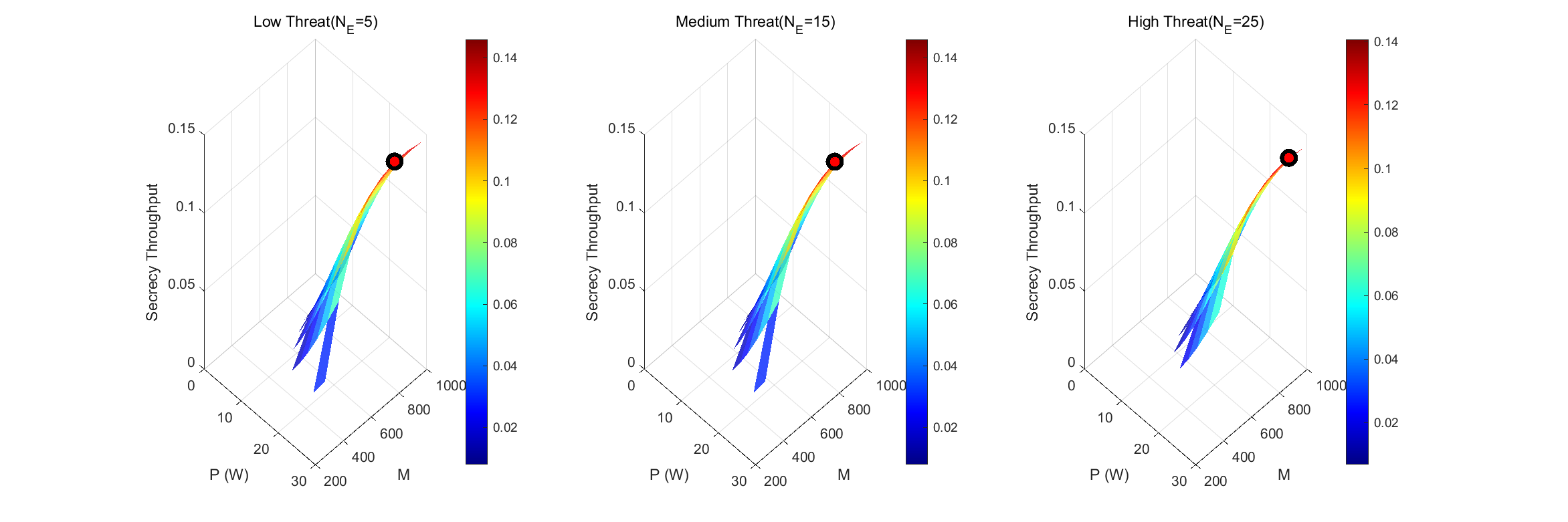}
		\caption{FAS Security Optimization Under Varying Threat Scenarios: 3D Secrecy Throughput Surface Analysis vs $P$ and $M$ for Different Eavesdropper Capabilities ($N_E = 5, 15, 25$).}\label{threat_ana3}
	\end{figure*}

	Figs.~\ref{threat_ana1}, \ref{threat_ana2}, and \ref{threat_ana3} present a comprehensive 3D surface analysis of the FAS security optimization under three eavesdropper threat levels: low ($N_E = 5$), medium ($N_E = 15$), and high ($N_E = 25$). The analysis is conducted using the GS algorithm with normalized antenna size $W = 3$, $m = 300$ bits, RU ports $N_R \in [5, 20]$, transmit power $P \in [0.1, 30]$~W, and blocklength $M \in [200, 1000]$. The black dot in each subplot marks the optimal operating point identified by the GS algorithm. The color gradient from blue to red represents the AAST values from low to high.
	
	Fig.~\ref{threat_ana1} depicts the secrecy throughput surface as a function of $P$ and $N_R$. Under the low-threat scenario ($N_E = 5$), the surface exhibits a broad high-performance plateau (red region) at large $N_R$ and moderate-to-high $P$, with the peak AAST reaching approximately $0.14$~bit/s/Hz. The optimal point is located at the corner of maximum $N_R$ and high $P$, consistent with Proposition~\ref{prop:monotone}. As the threat level increases to $N_E = 15$ and $N_E = 25$, the overall surface height decreases and the high-performance region contracts, but the peak AAST remains comparable ($\approx 0.14$). This indicates that while the eavesdropper's enhanced capability degrades the average throughput landscape, the GS algorithm can still locate an optimal $(N_R^*, P^*)$ that largely compensates for the increased threat. Furthermore, the surface gradient becomes steeper near the low-$N_R$ and low-$P$ boundary under high-threat conditions, highlighting that suboptimal parameter selection incurs a severe performance penalty in hostile environments.
	
	Fig.~\ref{threat_ana2} shows the secrecy throughput surface as a function of $M$ and $N_R$. For the low-threat case ($N_E = 5$), the surface reveals a clear ridge structure along the $N_R$ axis: the AAST first increases with $M$ and then decreases, exhibiting the unimodal blocklength behavior observed in Fig.~\ref{AASTvsM}. The peak AAST ($\approx 0.16$~bit/s/Hz) is achieved at large $N_R$ and an intermediate $M$, confirming that both port diversity and blocklength optimization are essential. As $N_E$ increases to $15$ and $25$, the surface contracts and the ridge shifts toward smaller $M$ values, implying that a shorter blocklength becomes preferable under stronger eavesdropper threats. The peak AAST decreases from approximately $0.16$ ($N_E = 5$) to $0.14$ ($N_E = 15$) and $0.13$ ($N_E = 25$), quantifying the throughput degradation caused by the escalating threat. The monotonically non-decreasing behavior of the AAST with respect to $N_R$ is consistently preserved across all three threat levels, further corroborating Proposition~\ref{prop:monotone}.
	
	Fig.~\ref{threat_ana3} illustrates the secrecy throughput surface as a function of $P$ and $M$, which reveals the most pronounced structural feature among the three views. The surface forms a narrow, elongated ridge in the $(P, M)$ plane: the AAST peaks at a specific combination of moderate $P$ and intermediate $M$, while it drops sharply on both sides of the ridge. For the low-threat case ($N_E = 5$), the ridge extends from approximately $(P, M) \approx (10, 300)$ to $(30, 400)$, with a peak AAST of approximately $0.14$~bit/s/Hz (marked by the black dot at the ridge summit). As $N_E$ increases to $15$ and $25$, the ridge narrows and the high-performance region (red/yellow) shrinks significantly, while the low-performance region (blue) expands. This narrowing effect demonstrates that under high-threat conditions, the system's tolerance for parameter mismatch diminishes considerably --- even small deviations from the optimal $(P^*, M^*)$ can lead to a dramatic throughput loss. The ridge structure also confirms the coupled nature of the $P$-$M$ trade-off: higher transmit power allows higher SNR and thus a broader range of feasible blocklengths, but excessive power provides diminishing returns due to the SNR-independent path-loss ratio $\bar{\gamma}_R/\bar{\gamma}_E$ discussed in Remark~\ref{remark:high_snr}.
	
	The combined analysis of Figs.~\ref{threat_ana1}--\ref{threat_ana3} leads to the following design-oriented insights.
	
	\begin{remark}[Threat-Adaptive Parameter Sensitivity]\label{remark:sensitivity}
		The 3D surface analysis reveals that the sensitivity of the AAST to parameter selection is highly asymmetric across the three optimization dimensions. From Figs.~\ref{threat_ana1} and \ref{threat_ana2}, the AAST is relatively robust to $N_R$ (monotonically non-decreasing, consistent with Proposition~\ref{prop:monotone}) and to $P$ (broad plateau at moderate-to-high $P$). However, from Figs.~\ref{threat_ana2} and \ref{threat_ana3}, the AAST is highly sensitive to the blocklength $M$, exhibiting a sharp peak-and-decay pattern. This asymmetry implies that in practical FAS deployments, \emph{blocklength selection is the most critical design parameter}: a suboptimal $M$ can degrade the throughput far more severely than a suboptimal $P$ or $N_R$. Therefore, system designers should prioritize accurate blocklength optimization while adopting simpler strategies (e.g., maximum $N_R$ and sufficient $P$) for the other parameters.
	\end{remark}

	\section{Conclusion}
	In this paper, we have investigated the physical layer security of a FAS-aided secure short-packet communication system under the VBCM framework. Closed-form and asymptotic expressions for the AAST have been derived by employing a linear approximation of the decoding error probability and Gauss-Chebyshev quadrature. Furthermore, a GS optimization algorithm has been developed to jointly optimize the transmit power, blocklength, and number of RU ports for maximizing the secrecy throughput. Numerical results have validated the accuracy of the analytical expressions and demonstrated that the FAS-aided system achieves significant secrecy performance gains over conventional fixed-position antenna systems. The results have also revealed the existence of an optimal blocklength and highlighted the sensitivity of system performance to parameter selection under varying eavesdropper threat levels.

\end{document}